\documentclass[11pt]{article}
\usepackage[a4paper, margin=1in]{geometry}
\usepackage{amsmath}
\usepackage{amsfonts}
\usepackage{amssymb}
\usepackage{amsthm}
\usepackage{mathrsfs}
\usepackage{mathtools}
\usepackage{graphicx}
\usepackage{microtype}
\usepackage{dsfont} 
\usepackage[utf8]{inputenc}
\usepackage[T1]{fontenc} 
\usepackage[ruled,vlined,linesnumbered]{algorithm2e}

\DeclareMathOperator*{\E}{\mathbb{E}}

\DeclareMathOperator{\dist}{dist}

\newcommand{\Q}{\mathcal{Q}}
\newcommand{\LSH}{\mathcal{H}}

\newtheorem{theorem}{Theorem}
\newtheorem{lemma}[theorem]{Lemma}

\theoremstyle{definition}
\newtheorem{definition}[theorem]{Definition}

\theoremstyle{remark}

\begin{document}

\title{Confirmation Sampling for Exact Nearest Neighbor Search}

\author{
	Tobias Christiani\footnote{The research leading to these results has received funding from the European Research Council under the European Union’s 7th Framework Programme (FP7/2007-2013) / ERC grant agreement no.~614331.}\\
	\small \texttt{tobc@itu.dk}\\
	\small IT University of Copenhagen and BARC
	\and
	Rasmus Pagh$^*$\footnote{Supported by Villum Foundation grant~16582 to Basic Algorithms Research Copenhagen (BARC). Part of this work was done while visiting Simons Institute for the Theory of Computing.} \\
	\small \texttt{pagh@itu.dk}\\
	\small IT University of Copenhagen and BARC
	\and
	Mikkel Thorup\footnote{Supported by an Investigator Grant from the Villum Foundation, Grant No.~16582.} \\
	\small \texttt{mikkel2thorup@gmail.com}\\
	\small University of Copenhagen and BARC
}
\maketitle

\begin{abstract}
Locality-sensitive hashing (LSH), introduced by Indyk and Motwani in STOC '98, has been an extremely influential framework for nearest neighbor search in high-dimensional data sets.
While theoretical work has focused on the \emph{approximate} nearest neighbor problems, in practice LSH data structures with suitably chosen parameters are used to solve the \emph{exact} nearest neighbor problem (with some error probability).
Sublinear query time is often possible in practice even for exact nearest neighbor search, intuitively because the nearest neighbor tends to be significantly closer than other data points.
However, theory offers little advice on how to choose LSH parameters outside of pre-specified worst-case settings.

We introduce the technique of \emph{confirmation sampling} for solving the exact nearest neighbor problem using LSH.
First, we give a general reduction that transforms a sequence of data structures that each find the nearest neighbor with a small, unknown probability, into a data structure that returns the nearest neighbor with probability $1-\delta$, using as few queries as possible.
Second, we present a new query algorithm for the \emph{LSH Forest} data structure with~$L$ trees that is able to return the exact nearest neighbor of a query point within the same time bound as an LSH Forest of $\Omega(L)$ trees with internal parameters specifically tuned to the query and data.
\end{abstract}

\section{Introduction}

Locality-sensitive hashing~\cite{indyk1998} (LSH) is the leading theoretical approach to nearest neighbor problems in high dimensions.
In nearest neighbor search we seek to preprocess a point set $P$ such that given a query point $q$, we can quickly return the point in $P$ that is closest to $q$ according to some distance measure $\dist(\cdot,\cdot)$.
Theoretical results are typically formulated as \emph{approximation} algorithms that allow a point at distance $cr$ to be returned if the nearest neighbor has distance~$r$ from the query point, where $c>1$ is a user-specified approximation factor.
In practice the quality parameter of interest is the \emph{recall}, i.e., the empirical probability of retrieving the nearest neighbor (see e.g.~\cite{andoni2015practical}).
As we will see below is not hard to show that LSH methods can obtain recall arbitrarily close to 1 if parameters are suitably chosen according to the given query and data set.
However, choosing parameters well, in an efficient way, is a challenge~\cite{lv2017intelligent}.

\subsection{Background}

\paragraph{Locality-sensitive hashing.}
A locality-sensitive family of hash functions $\LSH$ (an ``LSH family'') has the property that hash collision probability decreases as distance increases.
Specifically, for $h\sim\LSH$ the ``hash bucket'' $S_h(q) = \{ x\in P \mid h(x)=h(q) \}$ is more likely to contain the nearest neighbor of $q$ than any other element of $P$.
For a given data set $P$ one would typically use a family $\LSH$ such that the expected size of $S_h(q)$ is constant (for every $q$ or on average for a certain query distribution)~\cite{slaney2012optimal}.
Given such a family $\LSH$, suppose the nearest neighbor is $x_1\in P$, and define $p_1 = \Pr[x_1 \in S_h(q)]$ to be the probability of a hash collision with the nearest neighbor.
Then inspecting $S_{h_i}(q)$ for a sequence of hash functions $h_1,\dots,h_L$ independently sampled from $\LSH$ we will fail to find $x_1$ with probability $(1-p_1)^L \approx \exp(- p_1 L)$.
To make this as efficient as possible we can use a hash table that given $q$ allows us to retrieve $S_{h_i}(q)$ in time $O(1+|S_{h_i}(q)|).$
If we assume that the distance between $q$ and $x\in P$ can be computed in constant time, the expected time for this procedure is $O(L \E[1+|S_h(q)|])$.
There are several issues with the above construction:
\begin{itemize}
    \item If $p_1$ is large then the query algorithm still goes through $L$ hash buckets, even though we expect to see $x_1$ within the first $O(1/p_1)$ buckets.
    \item If $p_1 L$ is small, the recall $1-\exp(- p_1 L)$ is close to zero.
\end{itemize}
Notice that $p_1$ depends on the nearest neighbor that we are searching for, resulting in a chicken-and-egg situation: we would like to conduct the search with knowledge of $p_1$, but we only know $p_1$ if the search finds $x_1$ (and we know how the collision probability depends on $\dist(q, x_1)$).
We will introduce a technique called \emph{confirmation sampling} for dealing with the former problem of when to terminate the search when we have no knowledge of $p_1$.
The latter problem requires us to take a new look at how to query the so-called LSH forest data structure, described below.

\paragraph{Approximation versus recall.}
Early theoretical work on high-dimensional nearest neighbor search dealt with the simpler case of \emph{near neighbor} search where it is assumed that a maximum distance~$r$ to the nearest neighbor is known and a point within distance $cr$ must be returned.
A reduction with logarithmic overhead in time and space extends this to solve the approximate nearest neighbor problem with unknown distance $r$~\cite{indyk1998,har-peled2012}.
These reductions increase the approximation factor by $1+\gamma$, with space usage proportional to $1/\gamma$, and do not seem to provide any guarantee on recall even if used with a near-neighbor data structure with approximation factor~$c=1$.

A data structure known as \emph{LSH forest}, first described by Charikar~\cite{charikar2002} and later generalized and baptized by Bawa et al.~\cite{bawa2005lshforest}, removes the logarithmic overhead in space but the query algorithm still only provides $c$-approximate results and does not guarantee a specific recall.
Indeed, it is not hard to construct examples where there are many $c$-approximate nearest neighbors and the probability of returning the exact nearest neighbor is negligible.

\paragraph{LSH Forest.}
Since we will describe a new query algorithm for the LSH Forest data structure we review the data structure here.
We will again make use of an LSH family $\LSH$, but this family can be ``weak'' in the sense that collision probabilities are large, say, $Pr[h(q)=h(x)]=\Omega(1)$ for $x\in P$.
Assume for simplicity that we can sample $h\sim\LSH$ and evaluate $h(x)$ in constant time.
For parameters $K$ and $L$ and $(i,j)\in \{1,\dots,K\}\times\{1,\dots,L\}$, independently sample hash functions $h_{i,j} \sim \LSH$. Associate each point $x\in P$ with a string $h_j(x) = h_{1,j}(x) h_{2,j}(x) \dots h_{K,j}(x)$.
For $j=1,\dots,L$ the $j$th part of the LSH Forest is a trie that stores prefixes of the set of strings $h_j(P) = \{ h_j(x) \mid x\in P \}$.
Specifically, for each $x\in P$ it stores the shortest prefix of $h_j(x)$ that is unique among strings in $h_j(P)$ (if such a prefix exists, otherwise the whole string $h_j(x)$).
A pointer to $x$ is placed in the leaf corresponding to a prefix of $h_j(x)$.
The space for the data structure, not counting space for storing the $n$ points in $P$, is $O(nKL)$ words na\"ively, and can be improved to $O(nL)$ words using path compression~\cite{bawa2005lshforest}.

\paragraph{Querying LSH Forest.}
On a query $q$ and for a parameter $i\in\{1,\dots,K\}$, LSH Forest allows us to retrieve the hash bucket $S_{i,j}(q)$ of points in $P$ matching a length-$i$ prefix of $h_j(q)$ in time $O(i+|S_{i,j}(q)|)$.
We will use $p(q,x) = \Pr_{h\sim\LSH}[h(q)=h(x)]$ as shorthand for the collision probability between $q$ and $x$.
We have
\begin{equation}\label{eq:collisions}
    \E[|S_{i,j}(q)|] = \sum_{x\in P} p(q,x)^i \enspace .
\end{equation}
The larger the ``level'' $i$ is the smaller $S_{i,j}(q)$ is in expectation.
Conversely the probability of finding $x_1$ in the hash bucket is $\Pr[x_1 \in S_{i,j}(q)] = p(q,x_1)^i$ which decreases exponentially with~$i$.
The query algorithm described in~\cite{bawa2005lshforest} chooses the level $i_0$ to inspect as the smallest level where the number of collisions is linear, $i_0 = \min \{ i \mid \sum_{j=1}^L |S_{i_0,j}(q)| \leq cL\}$, for some constant $c$.
The probability of failing to find the nearest neighbor by inspecting all buckets $S_{i,1},\dots,S_{i,L}$ at level $i$ is $(1-p_1^i)^L \approx \exp(- p_1^i L)$, so to bound the failure probability we need to choose $L$ large enough.
For example, if the nearest neighbor of $q$ is in a dense cluster of $2cL$ points whose points almost surely reside in the same LSH bucket, the algorithm fails to find the nearest neighbor almost surely.
So LSH Forest is only ``self-tuning'' to a limited extent if high recall is desired: choosing a suitable parameter $L$ requires at least approximate knowledge of the distance distribution from $q$ to points of $P$.
Instead, we would like $L$ to be simply a parameter that determines the space usage, and use a different query algorithm that adapts to the data automatically.

\subsection{Our results}

LSH methods work by performing many iterations, each inspecting a hash table $\mathcal{D}_i$ with a small (and unknown) probability $p_1$ of finding the nearest neighbor.
It is easy to see that after $\ln(1/\delta)/p_1$ iterations the nearest neighbor will be retrieved with probability at least $1-\delta$.
We show that this number of iterations can be matched in expectation without knowledge of $p_1$, and in fact even without estimating any collision probabilities. 
Using a technique we call \emph{confirmation sampling} we obtain the following result on LSH-like methods:
\begin{theorem}\label{thm:confirmationLSH}
    Suppose there is a sequence of independent, randomized data structures $\mathcal{D}_1,\mathcal{D}_2,\dots$, such that on query $q$, $\mathcal{D}_i$ returns the nearest neighbor of $q$ in $P$ with probability \emph{at least} $p_q$ and each other point in $P$ with probability \emph{at most} $p_q$.
    Let $\delta > 0$ be given.
    There is an algorithm that depends on $\delta$ but not on $p_q$ that on input $q$ queries
    data structures $\mathcal{D}_1,\dots,\mathcal{D}_{j_q}$, performs $j_q$ distance computations, where $\E[j_q] = O(\ln(1/\delta)/p_1)$, and returns the nearest neighbor of $q$ with probability at least $1-\delta$.
\end{theorem}

Theorem~\ref{thm:confirmationLSH} shows that, at least in the case where we may use quadratic space to store a sufficiently long sequence of data structures $\mathcal{D}_i$, it suffices to focus on minimizing the product of the expected time for $\mathcal{D}_i$ and the number of iterations $1/p_1$.

In practice one would of course not have access to an unbounded sequence of data structures, but rather to a fixed number $L$ of data structures.
If these data structures offer a trade-off between query time and probability of returning the nearest neighbor it is still possible to apply Theorem~\ref{thm:confirmationLSH}:
For $i=1,2,\dots,\log n$ run confirmation sampling in rounds of $L$ steps with time budget $2^i$ for each data structure $\mathcal{D}_i$.
Terminate as soon as confirmation sampling returns a result --- by a union bound over the $\log n$ rounds the error probability is at most $\delta\log n$.

\medskip

Our second result addresses how to adapt not only to the collision probability of the nearest neighbor, but to the whole distance distribution from $q$ to points in $P$.
In particular, we design and analyze a new adaptive query algorithm for the LSH Forest data structure~\cite{charikar2002,bawa2005lshforest} discussed above.
LSH Forest is known to be able to adapt to the distance distribution to some extent, but previous work has required the query algorithm to depend on the distance to the nearest neighbor in $P$.
In contrast our query algorithm is independent of properties of the data.
The only requirement is that the LSH family used is \emph{monotone} in the sense that collision probability is non-increasing with distance.
We compare our adaptive algorithm to an optimal algorithm in a class of \emph{natural} algorithms that choose a level $i^*$ and a number of tries~$j^*$ (which may depend on the distance distribution between $q$ and $P$) and inspect the first $j^*$ buckets at level $i^*$.

\begin{theorem}\label{thm:adaptive}
    Let $OPT(L,K)$ denote the optimal cost of a natural algorithm that queries an LSH Forest data structure with $L$ trees and $K$ levels and returns the nearest neighbor with probability at least $1-1/n$.
    Further assume that the LSH family is monotone.
    Then there is an adaptive algorithm that queries an LSH Forest data structure with $O(L)$ trees and $K$ levels that returns the nearest neighbor using time $O(OPT(L,K))$ with probability $1-1/n$.
\end{theorem}

LSH Forest is not an asymptotically optimal data structure for approximate nearest neighbor search in general.
For example, it is known that data-dependent methods can be asymptotically faster in several important spaces, and data structures obtaining better space-time trade-offs are known~\cite{andoni2015optimal,andoni2017optimal}.
Generalizing our results for exact nearest neighbors to a data-dependent setting, say, in Euclidean space, is an interesting open direction.
Note that the data structures $\mathcal{D}_i$ in Theorem~\ref{thm:confirmationLSH} could be data dependent, though present data-dependent LSH techniques rely on knowing the (approximate) distance to the nearest neighbor.


\subsection{Related work}\label{sec:related}

There is a large literature on using LSH for nearest neighbors search in practice, often generalized to the $k$-nearest neighbor problem where the $k$ closest points in $P$ must be returned.
For simplicity we concentrate on the case $k=1$, but most results extend to arbitrary $k$.
Many heuristics that work well in practice come without guarantees on either result quality or query time in high dimensions, or provides guarantees only under certain assumptions on the data set.

\paragraph{Guarantees on recall.}
In practice, the performance of locality-sensitive hashing techniques is usually measured by their recall: the fraction of the true $k$-nearest neighbors found on average, see e.g.~\cite{andoni2015practical,aumuller2017annbenchmarks}.
From a theoretical point of view it is natural to bound the \emph{expected recall}, i.e., the probability that the nearest neighbor is found.
We are only aware of very few works that provide theoretical guarantees on expected recall in conjunction with sublinear query time in high dimensions and without assumptions on data.

Dong et al.~\cite{dong2008modeling} outline an ``adaptive'' method for achieving a given expected recall in the context of multiprobe LSH (with no formal statement of guarantees).
The idea is to determine, after inspecting $i$ buckets, whether to terminate or to inspect bucket $i+1$ based on the collision probability $p(q,\hat{x}_1)$ between $q$ and the nearest neighbor $\hat{x}_1$ found in the first $i$ buckets.
This requires an efficient method for computing $p(q,\hat{x}_1)$, which might not be known, especially for small collision probabilities.
This is not just a theoretical problem: Prominent LSH methods such as $p$-stable LSH~\cite{datar2004} and cross-polytope LSH~\cite{andoni2015practical} do not have closed-form expressions for collision probabilities.
Our adaptive algorithm is similar in spirit, but entirely avoids having to compute collision probabilities.

For the related \emph{near neighbor} problem where a search radius $r$ is given it is easier to give guarantees on recall, especially when collision probabilities at distance $r$ can be computed, see e.g.~\cite{christiani2017scalable}.

\paragraph{Parameter tuning.}
Since the performance of LSH data structures depends on parameter choices, a lot of work has gone into devising ways of choosing good parameters for a given data set, both during data structure construction and adaptively for the query algorithm.
Slaney et al.~\cite{slaney2012optimal} propose to select parameters based on the ``distance profile'' of a data set, but needs a bound on the distance to the nearest neighbor to function.

The state-of-the-art FALCONN library~\cite{andoni2015practical} uses grid search over parameters to empirically estimate the best parameters, assuming that the data and query distributions are identical.

We note that the adaptive method of Dong et al.~\cite{dong2008modeling} does not adapt search depth to the distance distribution from the query point $q$. 
In fact, choosing good parameters for LSH and especially multi-probe LSH was mentioned by Lv et al.~\cite{lv2017intelligent} as a challenge in the paper celebrating their VLDB 10-year Best Paper Award.


\section{Confirmation sampling}
Let $\mathcal{Q}$ denote a probability distribution with finite support $S$.
Further assume that elements of $S$ are equipped with a total ordering relation $\prec$, and define $x_1 = \min(S)$ as the smallest element in the support with respect to the ordering $\prec$.
Consider the problem of identifying $x_1$ given that we only have access to samples from the distribution $\mathcal{Q}$ and to the ordering, 
i.e., given elements $x,y\in S$ we can determine whether $x \prec y$, $x=y$, or $y \prec x$.
We propose a simple randomized algorithm for solving this problem that we call confirmation sampling.
The algorithm works by drawing samples from $\mathcal{Q}$ while keeping track of the smallest element seen so far together with the number of times it has been sampled in addition to the first sample --- the number of \emph{confirmations}. 
Once the smallest element has been confirmed $t$ times, the algorithm reports that element and terminates.
We use $\infty$ to denote an element that is larger than all elements of $S$.
\begin{algorithm}
\SetKwArray{Count}{count}
\DontPrintSemicolon
$\beta \leftarrow \infty$, \Count $\leftarrow 0$\;
\While{\Count $< t$}{
sample $X \sim \mathcal{Q}$\;
\uIf{$X = \beta$}
{\Count $\leftarrow$ \Count + 1\;\label{line:increase}}
\ElseIf{$X \prec \beta$}{
$\beta \leftarrow X$\;
\Count $\leftarrow 0$\;
}
}
\Return $\beta$\;
\caption{\textsc{ConfirmationSampling}$(\mathcal{Q}, t,\prec)$} \label{alg:confirmationsampling}
\end{algorithm}

\begin{theorem} \label{thm:cs}
Let $\mathcal{Q}$ denote a probability distribution with finite support $S$.
For $x_1 = \min(S)$ and $X\sim\mathcal{Q}$ let $p_1 = \Pr[X = x_1]$ and let $p_2 = \max \{ \Pr[X = x] \; | \; x\in S\backslash\{x_1\}\}$ be the largest sampling probability among elements of~$S$ other than $x_1$. Then:
      $$\Pr[\textsc{ConfirmationSampling}(\mathcal{Q}, t) \neq x_1] \leq (1-p_1) \left(\frac{ p_2 }{ p_1 + p_2 }\right)^{t} \enspace $$
The expected number of samples made by \textsc{ConfirmationSampling} is bounded by $(t+1)/p_1$.
\end{theorem}

Before we show Theorem~\ref{thm:cs} we observe that it implies Theorem~\ref{thm:confirmationLSH}:
Define an ordering on $P$ by $x \preceq_q y \Longleftrightarrow \dist(q,x) \leq \dist(q,y)$.
It can be turned into a total ordering $\prec_q$ by an arbitrary but fixed tie-breaking rule.
Choose $t=\lceil\log_2(1/\delta)\rceil$ and run \textsc{ConfirmationSampling}$(\mathcal{Q}, t,\prec_q)$ with the $i$th sample from $\mathcal{Q}$ being produced by querying $\mathcal{D}_i$ for the nearest neighbor of $q$. 
Since $p_1\leq p_q$ and $p_2\geq p_q$ we have that the error probability is bounded by $2^{-t} \leq \delta$.

\begin{proof}
If the algorithm fails to report $x_1$ it must have happened at least $t$ times that the confirmation counter was incremented (line~\ref{line:increase}) due to a sample $X$ satisfying the condition $X = \beta$ for $\beta \neq x_1$.
We will refer to such events as \emph{false confirmations} and proceed by upper bounding the probability that the algorithm performs $t$ false confirmations. 
Prior to each sample the probability of performing a false confirmation is maximized if $\beta = x_2$ for some $x_2\ne x_1$ maximizing the sampling probability, i.e., $\Pr[X = x_2] = p_2$.
Note also that the first sample can never result in a false confirmation.
The probability of the algorithm performing $t$ false confirmations before sampling $x_1$ can therefore be upper bounded by the probability that the first sample is not equal to $x_1$ and that we in the following samples observe $t$ samples of $x_2$ before sampling $x_1$. 
The probability that we sample $x_2$ conditioned on sampling either $x_1$ or $x_2$ is exactly $\frac{ p_2 }{ p_1 + p_2 }$, and the probability of this happening $t$ times in a row is $\left(\frac{ p_2 }{ p_1 + p_2 }\right)^{t}$.

To analyze the number of samples, consider an infinite sequence of independent samples $X_1,X_2,\dots \sim \mathcal{Q}$, and suppose that in the $i$th iteration the algorithm uses sample $X_i$.
Observe that the algorithm terminates no later than iteration $i$ if $x_1$ is sampled $t+1$ times in $X_1,\dots,X_i$.
The expected number of iterations needed to sample $x_1$ $t+1$ times is exactly $(t+1)/p_1$.
\end{proof}

Theorem \ref{thm:cs} is tight in the case where $\Q$ only assigns non-zero probability to two elements.
In Appendix~\ref{sec:exact} we derive the exact distribution of the output of \textsc{ConfirmationSampling}.
We observe that for the proof to work, the distribution from which samples are drawn does not need to be the same in each iteration of \textsc{ConfirmationSampling}, as long as $p_1$ is a lower bound on sampling $x_1$ and $p_2$ is an upper bound on sampling each element other than $x_1$.
If for some $\gamma \in [0,1]$ we have that every distribution satisfies $p_2/(p_1 + p_2) \leq \gamma$ then we can upper bound the error probability by $\gamma^t$. 
\subsection{Application to locality-sensitive hashing}
Assume that we have an LSH family that is tuned to give few collisions between query and non-neighbor points for a given query and data distribution.
Such a ``tuned'' LSH family may be obtained if the query distribution is known as discussed in section~\ref{sec:related}.
We can use confirmation sampling to adjust query time according to the distance to the nearest neighbor.

Let $(V, \dist)$ denote a distance space. That is, $V$ is equipped with a distance function $\dist \colon V \times V \to \mathbb{R}$.
We define locality-sensitive hashing~\cite{indyk1998} as follows:
\begin{definition}
	Let $\LSH$ denote a distribution over functions $h \colon V \to R$. 
	We say that $\LSH$ is locality-sensitive over $(V, \dist)$ if there exists a non-increasing $f \colon \mathbb{R} \to [0,1]$ such that for all $x, y \in V$ we have that 
	\begin{equation*}
		\Pr_{h \sim \LSH}[h(x) = h(y)] = f(\dist(x, y)).
	\end{equation*}
\end{definition}

We use the ordering $\prec_q$ defined above and define a distribution $\mathcal{Q}_q$ that is most easily described as a sampling procedure.
For now we will not care about the efficiency of implementing the sampling.
To create a sample $X \sim \mathcal{Q}_q$, sample $h \sim \LSH$, compute the ``bucket'' 
$$S(q) = \{ x \in P \mid h(x) = h(q) \}\enspace .$$
Now define $X$ as the element of $S(q)$ closest to $q$, if such an element exists, and otherwise a random element in $P$.%
\footnote{The sampling of a random element ensures compatibility with \textsc{ConfirmationSampling}, which requires a sample to be returned even if there is no hash collision. It is not really necessary from an algorithmic viewpoint, but also does not hurt the asymptotic performance.}
More precisely: If $S(q)\ne\emptyset$ we pick $X$ as the unique minimum element in $S(q)$ according to the total order $\prec_q$, and if $S(q)=\emptyset$ we pick $X$ uniformly at random from $P$.

\begin{lemma} \label{lem:conditional}
    For $X \sim \mathcal{Q}_q$ and any $x_2 \in P$, $\Pr[X = x_1] \geq \Pr[X = x_2]$.
\end{lemma}
\begin{proof}
    Since $\LSH$ is locality-sensitive we have that $\Pr[h(q)=h(x_1)] \geq \Pr[h(q)=h(x_2)]$.
    Thus
    \begin{align*}
        \Pr[X=x_1] &= \Pr[h(q)=h(x_1)] + \frac{\Pr[S(q)=\emptyset]}{n} \\
        &\geq \Pr[h(q)=h(x_2)] + \frac{\Pr[S(q)=\emptyset]}{n} \\
        &= \Pr[X=x_2] \enspace .
    \end{align*}
\end{proof}

As before Theorem~\ref{thm:cs} now implies that confirmation sampling succeeds with good probability:

\begin{lemma} \label{lem:cslsh}
    Let $x_1$ be the nearest neighbor of $q$ in $P$ (breaking any ties according to $\prec_q$).
    \textsc{ConfirmationSampling}$(\mathcal{Q}_q,t,\prec_q)$ returns $x_1$ with probability least $1 - 2^{-t}$.
    The expected number of samples from $\mathcal{Q}_q$ is bounded by $(t+1)/p_1$, where $p_1 \geq \Pr[h(q)=h(x_1)]$.
\end{lemma}

To efficiently sample from $\mathcal{Q}_q$ we independently sample $h_1, h_2,\dots \sim \LSH$, and construct a sequence of hash tables $\mathcal{D}_1,\mathcal{D}_2,\dots$ that allow us to find $S_i(q) = \{ x \in P \mid h_i(x) = h_i(q) \}$ in time $O(1+|S_i(q)|)$.
Random samples from $P$ can be realized using an array of pointers to elements of~$P$.

We note that the above is not an entirely satisfactory solution, since the number of data structures needed cannot be bounded ahead of time (or rather, $\Omega(n)$ data structures may be needed to succeed, resulting in quadratic space usage).
A possible remedy if the algorithm does not terminate after inspecting $L$ hash tables is \emph{multi-probing}~\cite{panigrahy2006,lv2017intelligent} where more than one bucket is inspected in each hash table.
Multiprobing increases the probability $p_1$ of finding the nearest neighbor in each hash table.
In the next section we consider another approach to dealing with a space-bounded data structure.
\section{Fully adaptive nearest neighbor search}
We present an adaptive algorithm for nearest neighbor search in an LSH Forest that succeeds with high probability\footnote{For every choice of constant $c \geq 1$ there exists a constant $n_0$ such that for $n \geq n_0$ we can obtain success probability $1 - 1/n^c$ where $n = |P|$ denotes the size of the set of data points.} and matches the minimum expected running time that can be obtained by a natural algorithm that has full knowledge of the LSH collision probabilities between the query point and all the data points, provided we are are allowed a constant factor increase in the number of trees used by the algorithm.
We define $OPT(L, K)$ as the minimum expected search time that can be achieved by an algorithm with access to an LSH Forest of $L$ trees of depth $K$ where the algorithm can choose to search $j \leq L$ trees at level $i \leq K$ with the requirement that the nearest neighbor should be reported with probability at least $1 - 1/n$.
\begin{equation*}\label{eq:opt}
    OPT(L, K) = \min \{ (\ln n)  (i + \sum_{x \in P} p(q, x)^i)/p(q, x_1)^i \mid 0 \leq i \leq K,\, p(q, x_1)^i L \geq \ln n \}
\end{equation*}
We note that $OPT(L, K)$ only reflects the optimal running time under the assumption that $p_1$ is bounded away from $1$. 
If for example we had $p_1 = 1$ the multiplicative overhead of $\ln n$ in the running time would not be needed. 

\paragraph{Overview of our approach.}
The algorithm works by measuring the number of collisions at different levels in the LSH Forest and with high probability adapting to search at a level that will result in $O(OPT(L, K))$ running time.
Ideally, given sufficiently many trees, we would like to search the level $i$ that balances the number of hash function evaluations and the expected number of collisions with the query point.
However such a level might not exist as the expected number of collisions can decrease by more than a constant factor as we increase the level.

We begin by introducing some notation. Let $p_1 = p(q, x_1)$, where $x_1$ denotes the nearest neighbor to $q$ in $P$ and define:
\begin{equation*}
	C(i) = \sum_{x \in P} p(q, x)^{i}, 
\end{equation*}
\begin{equation*}
	T(i) = (i + C(i))/p_{1}^{i}. 
\end{equation*}
Observe that $C(i)$ is the expected number of collisions with the query point at level $i$, and $T(i)$ is the expected running time of an algorithm that searches at level $i$ and guarantees reporting the nearest neighbor of $q$ with some constant probability.
If we let $i^*$ denote the choice of level resulting in the minimum value of $OPT(L, K)$ then $OPT(L, K) = T(i^*) \ln n$. 
Finally, define $i'$ to be the smallest integer $i$ such that $C(i) \leq i$. 

Given that the number of trees $L$ is sufficiently large we can show that searching either the first $1/p_{1}^{i'}$ trees at level $i'$ or the first $1/p_{1}^{i' - 1}$ trees at level $i'- 1$ results in an expected running time that is bounded by $O(T(i^*))$ while we report the nearest neighbor with constant probability at least $1 - 1/e$.
That is, one of the two levels right around where the number of hash function evaluations and the number of collisions balance out (we have $C(i') \leq i'$ and $C(i'-1) > i' - 1$) result in optimal running time for constant failure probability. 
Since we don't know $p_1$ we can search \emph{both} of these levels using confirmation sampling, in parallel, until one of them terminates. 
This gives us an algorithm that with constant probability terminates in time $O(T(i^*))$ and reports the nearest neighbor.
In order to reduce the failure probability to $1/n$ while obtaining optimal running time in the high probability regime we can perform $O(\log n)$ independent repetitions, so that conceptually there are $O(\log n)$ independent forests, and stop the search once a constant fraction terminates. 

\paragraph{Query algorithm and parameters.}
There are two circumstances that prevent us from being able to use the approach outlined above.
The primary problem is that we don't know the value of $i'$ and estimating it appears to be difficult. 
The solution proposed by our algorithm is to instead search the ``empirical'' $i'$ and $i'-1$:
we measure the number of collisions at different levels and search level $i$ and $i-1$ where $i$ is set to the minimum level where the average number of collisions is smaller than $i$.
This procedure is described in pseudocode in the for-loop section of Algorithm~\ref{alg:ANN}.

The second problem is that restrictions on $L$ and $K$ can make it necessary for us to search a level $i < i' -1$, either because $K < i' - 1$ or because $L$ is too small to ensure that we find the nearest neighbor by searching at level $i' - 1$. 
The second part of Algorithm~\ref{alg:ANN} that runs when $j = L'$ deals with this problem by searching through the LSH forests bottom-up until a level that results in optimal running time is encountered.

\begin{algorithm}
\DontPrintSemicolon
\For{$j \gets 1,2,4,\dots, L'$}{
    find the smallest level $i$ such that the first $j$ trees in at least half of the forests have at most $10ij$ collisions. If such a level does not exist set $i = K$.\;\label{line:level}
    in each forest run confirmation sampling at level $i$ and $i-1$ with a time budget of $10ij$ (looking at no more than $j$ buckets and at no more than $10ij$ collisions).\;
	\If{confirmation sampling terminated in $1/4$ of the forests at level $i$ or level $i-1$} {
	   report the closest point seen so far and terminate.
	}
	\If{$j = L'$} {
	   run confirmation sampling in lock-step across the forests starting at level $i-1$, decreasing the level and starting over once $1/2$ of the searches have explored tree number $L'$. Do this until confirmation sampling terminates in $1/4$ of the forests.
	}
}
\caption{\textsc{AdaptiveNearestNeighbor}$(q)$} \label{alg:ANN}
\end{algorithm}

We aim for matching the running time of $OPT(L, K)$ up to constant factors when we are allowed to use $O(L)$ trees.
Algorithm \ref{alg:ANN} operates on $\Theta(\log n)$ LSH Forests that each has $L'$ trees where $L' = O(L / \log n)$ is a sufficiently large power of two. 
The confirmation sampling used to search in these forests has a parameter setting of $t = 3$ since we only need each search to terminate and correctly report the nearest neighbor with a sufficiently large constant probability. 

The proof of Theorem \ref{thm:adaptive} is based on two arguments.
First we will show that the stopping condition that $1/4$ of the forests at a given level terminates within the time budget ensures that the nearest neighbor is always reported with high probability.
Second, we show that with high probability the algorithm terminates in time $O(OPT(L, K))$. 

\paragraph{Correctness.}
The choice of $i$ made by the algorithm always satisfies $i \leq n$ since there can be no more than $n$ collisions at any level.
If we show correctness with high probability at a fixed level then we can use a simple union bound over the first $n$ levels to show that with high probability at every level where $1/4$ of the searches terminate we have found the nearest neighbor of the query point.
The instances of confirmation sampling used by Algorithm \ref{alg:ANN} use $t = 3$ confirmations before terminating. 
According to Lemma \ref{lem:cslsh} the probability of terminating and reporting a point different from the nearest neighbor is at most $1/8$. 
By applying a standard Chernoff bound we can show that over $O(\log n)$ independent runs of confirmation sampling with high probability less than $1/4$ the instances will fail to report the nearest neighbor. 

\paragraph{Bounding the running time.}
We remind the reader that we use $i^{*}$ to denote the underlying choice of level that minimizes $OPT(L, K)$, that $i'$ denotes the minimum level such that $C(i') \leq i'$, and that $i$ is the choice of level made by the query algorithm.

Consider line \ref{line:level} of Algorithm~\ref{alg:ANN} where the level $i$ is set to the smallest level where the first $j$ trees in at least half the forest have at most $10ij$ collisions. 
This operation can be completed in $O(ij)$ time per forest by proceeding top-down across all the forests and for each forest summing up the number of collisions across all its tries at the current level until level $i$ is reached. 
We make use of constant-time access to the size of buckets/subtrees as we search down in an LSH Forest trie (either by explicitly storing the size of subtrees when we construct the trie, or by inspecting the pointers to the bucket associated with a given prefix).

We will now argue that with high probability Algorithm~\ref{alg:ANN} terminates in time $O(OPT(L, K))$ in each of the two following cases:
%
%
\paragraph{Case 1: $C(i^{*}) \leq i^{*}$.}
We will show that there exists a value of $j \leq L'$ such that with high probability the algorithm terminates at this value (or earlier) and in $O(OPT(K, L))$ time.
Consider the first iteration of the for-loop where $100/p_1^{i^{*}} \leq j \leq L'$.
Such a $j$ exists by the restrictions underlying the choice of level that minimizes $OPT(L, K)$ and by our freedom to set $L' = O(L / \log n)$.
By Markov's inequality the probability that the number of collisions in the first $j$ trees of a forest at level $i'$ is greater than $10i'j$ is at most $1/10$.
Therefore it happens with high probability that the algorithm sets $i \leq i' \leq i^{*}$ where the last inequality follows from the definition of $i'$ and the assumption that $C(i^{*}) \leq i^{*}$.
By our choice of $j$ we know that confirmation sampling at level $i$ will terminate in each forest with a large constant probability, say, $9/10$. 
With high probability we therefore have that in at least $1/4$ of the forests confirmation sampling at level $i$ terminates within the budget of $10ij$.
To bound the total running time we use that with high probability $i \leq i'$ for every value of $j$ and since $j$ is doubled at every step of the for loop we can bound the running time in all $O(\log n)$ LSH forests by $O(i'j \log n) = O(T(i^{*}) \log n) = O(OPT(L, K))$.

\paragraph{Case 2: $C(i^{*}) > i^{*}$.}
Consider first the sub-case where $i^{*} = i' - 1$.
Suppose there exists a minimum $j \leq L'$ such that $i'j \geq 100\, T(i' -1)$, $j$ is an integer power of 2, and $j \geq 100/p_1^{i'-1}$ (the latter condition holds by the assumption $i^* = i' - 1$). 
We previously argued that with high probability the algorithm sets $i \leq i'$.
In the first iteration of the for-loop where $j$ takes on this value the following holds:
If $i = i'$ then level $i'-1$ is searched with a sufficiently large budget to ensure termination with high probability.
If $i < i'$ then level $i'-1$ is searched up until tree number $j$, again ensuring termination with high probability. 
In both of these cases the running time is bounded by $O(OPT(L, K))$.
Otherwise, if $L'i' < T(i' - 1) / 100$ then with high probability the time spent in the for-loop part of the algorithm is upper bounded by $O(T(i' - 1) \log n) = O(OPT(L, K))$, and if level $i' - 1$ was not searched in the for-loop then it will be searched in the first step of the bottom-up part of the algorithm (because $i \leq i'$ with high probability) where we are guaranteed to terminate in optimal time with high probability.

Consider now the sub-case where $i^* < i' - 1$. 
Let $\hat{i}$ denote the largest level satisfying $\hat{i} < i' - 1$ and $100/p_1^{\hat{i}} \leq L'$.
The query algorithm will terminate with high probability when having searched sufficiently many trees at level $\hat{i} \geq i^*$.
We will proceed by bounding the cost up to the point where $100/p_1^{\hat{i}}$ trees have been searched in half of the forests at level $\hat{i}$.
The cost of running the for-loop part of the algorithm is bounded by $O(L'i' \log n)$ with high probability.
The number of collisions encountered through the bottom-up search when having searched level $\hat{i} + 1$ is with high probability bounded by $O(C(\hat{i} + 1)L' \log n ) = O((C(\hat{i} + 1)/p_1^{\hat{i} + 1}) \log n)$ since $100/p_1^{\hat{i} + 1} > L'$ by our choice of $\hat{i}$.
Finally, the cost of searching at level $\hat{i}$ until $1/4$ of the forests terminate is bounded by $O(T(\hat{i}) \log n)$ with high probability.

Next we show that the sum of all these costs is bounded by $O(OPT(L, K))$.
For every $x \in P$ it holds by monotonicity that $p_1 = p(x_1, q) \geq p(x, q)$ and it follows that for every $i$ we have $C(i+1) \leq p_1 C(i)$.
Applying this inequality we get the bound $C(\hat{i} + 1)/p_1^{\hat{i} + 1} \leq C(i^*)/p_1^{i^*} \leq T(i^*)$ that is used to bound the number of collisions from the bottom-up search.
The same approach also gives a bound on the number of collisions at level $\hat{i}$.
In order to bound the contribution from the for-loop note that $C(\hat{i} + 1) \geq C(i' -1) > i' - 1$ where the last inequality holds by the definition of $i'$.
It also holds that $L' < 100/p_1^{\hat{i} + 1}$ by the choice of $\hat{i}$.
Combining these two inequalities  $L'i' \leq 100\,(C(\hat{i} + 1) + 1)/p_1^{\hat{i} + 1} = O(T(i^*))$.
The bound on the total running time is then given by $O(T(i^*) \log n)) = O(OPT(L, K))$.
\section{Conclusion and open problems}
We have introduced confirmation sampling as a technique for identifying the minimum element from a discrete distribution.
Confirmation sampling works particularly well when the minimum element is at least as likely to be sampled as other elements.
Combining confirmation sampling with locality-sensitive hashing we obtain a randomized solution to the exact nearest neighbor search problem that works without knowledge of the probability of collision between pairs of points.
We use these techniques to design a new adaptive query algorithm for the LSH Forest data structure with $L$ trees that returns the nearest neighbor of a query point with the same time bound that is achieved if the query algorithm has access to an LSH forest of $\Omega(L)$ trees with internal parameters specifically tuned to the query and data. 

We can use confirmation sampling with LSH to solve the $k$-nearest neighbor problem with high probability in $k$ by keeping track of the top-$k$ closest points and requiring each to be confirmed $O(\log k)$ times.
If we are able to compute the collision probabilities we can use the adaptive stopping rule of Dong et al.~\cite{dong2008modeling} to stop the search once we have sampled $j \geq \ln(1/\delta)/\hat{p}_k$ buckets, where $\hat{p}_k$ is the collision probability between the query point and the $k$th nearest neighbor candidate found by the query algorithm. 
This stopping rule guarantees that if $x$ is a $k$-nearest neighbor to the query point, and the LSH family is monotone, then $x$ is reported with probability at least $1-\delta$.
It would be interesting to find a similarly efficient stopping rule for $\delta = \Theta(1)$ that works without knowledge of the collision probabilities.  

Our adaptive query algorithm for the LSH Forest data structure makes use of union bounds over the $K$ levels of the data structure when showing correctness and also uses that with high probability it doesn't search too far (something which could potentially cost time $O(n)$). When we compare our performance against an optimally tuned algorithm that must succeed with high probability we can afford to pay for this extra overhead. It remains an open problem to find an adaptive query algorithm that matches an optimally tuned algorithm when we only require constant success probability, even if we can compute collision probabilities

\newpage

\appendix
\section{Exact distribution of the output of \textsc{ConfirmationSampling}}\label{sec:exact}
Suppose $S=\{x_1,\dots,x_n\}$, where indices are chosen such that $p_i = \Pr[X=x_i]$ is non-decreasing in $i$: $p_1\geq p_2 \geq \dots \geq p_n$.
Given a distribution $\Q$ and a parameter $t$ let $\varrho_i$ denote the probability that \textsc{ConfirmationSampling}$(\Q, t)$ reports element $x_i$.
Consider an infinite sequence of i.i.d.\ samples $X_1, X_2, \dots$ from $\Q$.
If $x_i$ is sampled $t + 1$ times before a single sample of $x_j$ with $j < i$ then the algorithm reports $x_i$.
It is easy to see that $\varrho_n = p_n^{t+1}$ since the only way that $x_n$ gets reported is if the first $t+1$ samples $X_1, \dots, X_{t+1}$ are equal to $x_n$.
We can extend this idea to obtain the expression for $\varrho_i$.
\begin{lemma} \label{lem:exact}
	\begin{equation*}
		\varrho_i = \left(1 - \sum_{j > i}\varrho_j\right)\left(\frac{p_i}{\sum_{s \leq i} p_s}\right)^{t+1}.
	\end{equation*} 
\end{lemma}
\begin{proof}
	We will gradually reveal information about the outcomes of the sequence $X_1, X_2, \dots$ in order to arrive at the expression in the Lemma.
	We begin by asking the question for each $X_{i}$ whether $X_i = x_n$ or $X_i < x_n$.
	Only if the first $t+1$ samples are equal to $x_n$ does the algorithm report $x_n$.
	Otherwise we can restrict our attention to the elements $X_i < x_n$ and ask the same question for $x_{n-1}$ and so on.
\end{proof}
For a specific choice of distribution we can compare the exact probability that confirmation sampling fails to report the minimum element with our upper bound in Theorem \ref{thm:cs}. From inspection: if we consider the uniform distribution the failure probability appears identical for $t = 1$ and as we increase $t$ the upper bound is at most twice as large as the actual failure probability.
\bibliography{knn}
\bibliographystyle{abbrv}
\end{document}